%% file: main.tex
\documentclass[aps,pra,reprint,superscriptaddress,nofootinbib]{revtex4-2}
\usepackage[utf8]{inputenc}
\usepackage[T1]{fontenc}
\usepackage{lmodern,microtype}
\usepackage{amsmath,amssymb,amsfonts,amsthm,mathtools,bm,mathrsfs}
\usepackage{booktabs,enumitem,graphicx,xcolor}
\usepackage{silence}
\usepackage[colorlinks=true,linkcolor=blue!45!black,citecolor=blue!45!black,urlcolor=blue!45!black]{hyperref}
\usepackage[nameinlink,capitalize]{cleveref}
\graphicspath{{figures/}}
\hypersetup{
 pdftitle={Maximal Classicalization of Finite-Group Quantum Reference-Frame Noise},
 pdfauthor={Maxim V. Churilov},
 pdfsubject={Quantum reference frames, random-unitary channels, and representation theory},
 pdfkeywords={quantum reference frames, random-unitary channels, diamond norm, finite groups, representation theory}
}

\newtheorem{theorem}{Theorem}[section]

\newtheorem{corollary}[theorem]{Corollary}
\newtheorem{definition}[theorem]{Definition}
\newtheorem{remark}[theorem]{Remark}
\newtheorem{example}[theorem]{Example}

\newcommand{\CPTP}{\mathsf{CPTP}}
\newcommand{\id}{\mathrm{id}}
\newcommand{\Tr}{\operatorname{Tr}}
\newcommand{\Ad}{\operatorname{Ad}}
\newcommand{\Prob}{\mathsf P}
\newcommand{\TV}{d_{\mathrm{TV}}}
\newcommand{\Irr}{\operatorname{Irr}}
\newcommand{\supp}{\operatorname{supp}}
\newcommand{\rank}{\operatorname{rank}}
\newcommand{\ketbra}[2]{|#1\rangle\!\langle#2|}

\input{physics_format}

\begin{document}

\title{Maximal Classicalization of Finite-Group Quantum Reference-Frame Noise}
\author{Maxim V. Churilov}
\email{churilovm1305@gmail.com}
\affiliation{Independent Researcher, Orenburg, Russia}
\date{April 5, 2026}

\begin{abstract}
A finite quantum reference token with group-valued misalignment induces a random-unitary channel.  We study simulation by a single arbitrary CPTP map applied after that channel.  For a unitary representation $U$ of a finite group $G$, we prove that the following conditions are equivalent: $U$ contains every irreducible type; one ancilla-assisted input has an orthonormal $G$-orbit; signed group measures embed isometrically into channels in diamond norm; and, for every pair of noise laws $p,q$,
\[
 \inf_{\Lambda\in\CPTP}\frac12\|\Phi_q^U-\Lambda\Phi_p^U\|_\diamond
 =\min_{r\in\Prob(G)}\frac12\|q-r*p\|_1.
\]
Thus representation completeness is the exact carrier condition for universal reduction of quantum post-processing to classical convolution.  We determine the minimum ancilla dimensions for an orthogonal orbit and for an invariant calibration seed.  For an incomplete carrier, with visible Plancherel dimension $S(U)$, we derive the exact conditional-expectation distance $\frac12\|\id-\Phi_u^U\|_\diamond=1-1/S(U)$ and an explicit quantum--classical deficiency gap.  For irreducible carriers the deficiency is obtained in closed form; the faithful two-dimensional representation of $S_3$ yields an exact ten-percent reduction relative to classical convolution.  We also characterize law identifiability through the conjugation representation, provide finite linear programs and decision witnesses, and establish both a finite-dimensional obstruction and stable visible-band reconstruction for infinite compact groups.  Deterministic ancillary code reproduces the finite-group examples and numerical regression checks.
Universal classicalization is inherited by every subquotient, yielding a
monotone obstruction hierarchy and corresponding quantitative lower
bounds.  Exact orbit orthogonality is supplemented by a
robust frame theorem: the smallest eigenvalue of an experimentally prepared
orbit Gram matrix gives a simultaneous multiplicative lower bound for the
diamond norm of every signed branch contrast.
\end{abstract}

\maketitle

\section{Introduction}

A noisy finite quantum reference frame is often described by an unknown group element $g\in G$.  If the token transforms under a unitary representation $U:g\mapsto U_g$ and the relative displacement has law $p$, discarding the classical branch produces
\begin{equation}
 \Phi_p^U(X)=\sum_{g\in G}p(g)U_gXU_g^\dagger .
 \label{eq:random-unitary}
\end{equation}
Given a source law $p$ and a target law $q$, the directional post-processing error is
\begin{equation}
 \delta_U(q|p)=\inf_{\Lambda\in\CPTP}
 \frac12\|\Phi_q^U-\Lambda\circ\Phi_p^U\|_\diamond .
 \label{eq:quantum-deficiency}
\end{equation}
Although $p$ and $q$ are classical, the optimization is not: $\Lambda$ may be coherent, may break the symmetry, may use an environment of arbitrary dimension, and is tested on inputs entangled with an ancilla.  Throughout, $\Lambda$ is a one-shot output post-processing channel on the same carrier.  The results do not include preprocessing, adaptive access to several uses, or general superchannels.

A classical correction consists of sampling an additional group element and composing the corresponding translation.  With the convention
\begin{equation}
 (r*p)(h)=\sum_{g\in G}r(hg^{-1})p(g),
 \label{eq:convolution}
\end{equation}
one has $\Phi_r^U\Phi_p^U=\Phi_{r*p}^U$, and therefore
\begin{equation}
 \delta_U(q|p)\le
 \delta_{\rm cl}(q|p):=
 \min_{r\in\Prob(G)}\frac12\|q-r*p\|_1.
 \label{eq:classical-deficiency}
\end{equation}
The central question is when equality holds for every $p,q$.

For the regular representation, an entangled probe resolves all branches and classicalization is plausible.  Regularity, however, demands carrier dimension $|G|$.  We show that the exact boundary is smaller and representation theoretic: one copy of each irreducible type is necessary and sufficient.  Sufficiency follows from an invariant weighted-character state whose orbit is orthogonal.  The same orbit both realizes the $\ell_1$ norm of every signed group measure and constrains an arbitrary local post-processing to induce a subnormalized convolution kernel.  Necessity is more delicate.  Assuming universal equality in \cref{eq:classical-deficiency}, we perturb the uniform law in the direction that separates the identity branch from the uniform mixture of all nonidentity branches.  Equality forces diamond distance two, hence an orthogonal orbit, and therefore the presence of every irreducible type.

The converse identifies the exact boundary between complete and incomplete carriers.  In particular, faithful representations can fail universal classicalization even when the map $p\mapsto\Phi_p^U$ is injective on the relevant family.  The obstruction is not only loss of a Fourier mode; it can be the absence of one common input that resolves all branches.

\subsection{Relation to prior work and claim boundary}

Quantum reference frames, asymmetry, and group twirling are established subjects \cite{Bartlett2007,GourSpekkens2008,MarvianSpekkens2014,Chiribella2004}.  Random-unitary channels and ancilla-assisted unitary discrimination are likewise well developed \cite{AudenaertScheel2008,Rosgen2008,ChenYing2010,Bavaresco2022,Watrous2018,Holevo2012}.  General channel comparison is governed by quantum randomization theorems and diamond deficiency \cite{Shmaya2005,Buscemi2012,Jencova2016}, while environment-seizable families provide important cases in which channel discrimination reduces to state discrimination \cite{Wilde2019}.  Twirling-channel zero-error structure has also been related explicitly to irreducible multiplicities \cite{LiuHan2023}, and group-covariant channel families have been studied from a complementary extremal perspective \cite{MemarzadehSanders2022}.

The result established here is not the character orthogonality identity
by itself, nor a new general theory of covariant channels.  It is the
equivalence between representation completeness and
\emph{universal unrestricted post-processing classicalization},
including a strict converse for every incomplete carrier.  The exact
visible-support diamond distance, the irreducible deficiency, and the
finite decision certificates are consequences of this classification.
The gap for an incomplete carrier is a property of its quantum channel
encoding and need not require a coherent optimal converter; in the
uniform perturbation used below, the identity converter already
witnesses the strict inequality.

The conditional-expectation estimate used to evaluate
\cref{eq:id-twirl-distance} is a finite-dimensional
Pimsner--Popa-index argument \cite{PimsnerPopa1986}.  We give a self-contained block proof and use it only for the conditional expectation generated by the representation twirl.

\paragraph{Status of the main ingredients.}
Character orthogonality and ancilla-assisted unitary discrimination underlie the orbit criterion; quantum randomization and covariance underlie the sufficiency argument.  The principal statements are the converse showing that universal classicalization forces complete irreducible support, and the resulting quantitative failure bounds for incomplete carriers.

\section{Finite-group setting}

Let $G$ be a finite group of order $N>1$.  We use standard finite-group representation theory throughout \cite{Serre1977,FultonHarris1991}.  Write
\begin{equation}
 \mathcal H\simeq\bigoplus_{\lambda\in\Irr(G)}
 V_\lambda\otimes\mathbb C^{m_\lambda},
 \qquad
 U_g\simeq\bigoplus_\lambda U_g^{(\lambda)}\otimes I_{m_\lambda}.
 \label{eq:isotypic}
\end{equation}

\begin{definition}[Representation-complete carrier]
The representation $U$ is \emph{representation complete} if $m_\lambda\ge1$ for every irreducible representation $\lambda$ of $G$.
\end{definition}

The regular representation has multiplicity $d_\lambda=\dim V_\lambda$.  Representation completeness requires only one copy and can therefore be much smaller.

For a complex function $\sigma:G\to\mathbb C$, define the linear map
\begin{equation}
 \Phi_\sigma^U=\sum_{g\in G}\sigma(g)\Ad_{U_g}.
 \label{eq:signed-map}
\end{equation}
For probability laws this is \cref{eq:random-unitary}; for signed laws it need not be positive.

\section{Maximal classicalization theorem}

\begin{theorem}[Equivalent forms of complete group visibility]
\label{thm:main-equivalence}
For a finite-dimensional unitary representation $U$ of $G$, the following are equivalent.
\begin{enumerate}[label=(\roman*),leftmargin=2.2em]
\item $U$ is representation complete.
\item There exists a density operator $\rho$ commuting with every $U_g$ such that
\begin{equation}
 \Tr(\rho U_g)=\delta_{g,e}.
 \label{eq:delta-character}
\end{equation}
\item There exist a finite ancilla $R$ and a unit vector $|\Psi\rangle\in\mathcal H\otimes R$ whose orbit
\begin{equation}
 |\Psi_g\rangle=(U_g\otimes I_R)|\Psi\rangle
 \label{eq:orbit}
\end{equation}
is orthonormal.
\item For every complex function $\sigma$ on $G$,
\begin{equation}
 \|\Phi_\sigma^U\|_\diamond=\|\sigma\|_1.
 \label{eq:diamond-isometry}
\end{equation}
\item For every pair of probability laws $p,q$,
\begin{equation}
 \boxed{\delta_U(q|p)=\delta_{\rm cl}(q|p)}.
 \label{eq:universal-classicalization}
\end{equation}
\end{enumerate}
An ancilla in (iii) can be chosen with dimension at most $\sum_\lambda d_\lambda$.
\end{theorem}

\begin{proof}
Assume (i).  Choose one unit vector $|a_\lambda\rangle$ in every nonzero multiplicity space and set
\begin{equation}
 \rho=\bigoplus_{\lambda\in\Irr(G)}
 \frac{d_\lambda}{N}I_{V_\lambda}\otimes
 \ketbra{a_\lambda}{a_\lambda}.
 \label{eq:rho}
\end{equation}
Since $\sum_\lambda d_\lambda^2=N$, this is a density operator.  It commutes with $U$, and the regular-character identity gives
\begin{equation}
 \Tr(\rho U_g)=\frac1N\sum_\lambda d_\lambda\chi_\lambda(g)=\delta_{g,e}.
\end{equation}
Thus (i)$\Rightarrow$(ii).  A canonical purification of $\rho$ satisfies
\begin{equation}
 \langle\Psi_g|\Psi_h\rangle
 =\Tr(\rho U_g^\dagger U_h)=\delta_{g,h},
\end{equation}
proving (ii)$\Rightarrow$(iii).  The rank of \cref{eq:rho} is $\sum_\lambda d_\lambda$.

If (iii) holds, the span of the orbit is invariant and $U\otimes I_R$ acts on it as the left regular representation.  Tensoring by a trivial ancilla changes multiplicities but not the set of irreducible types, so every irreducible type already occurs in $U$.  Hence (iii)$\Rightarrow$(i).

For (iii)$\Rightarrow$(iv), apply $\Phi_\sigma^U\otimes\id_R$ to $\ketbra{\Psi}{\Psi}$.  The output is diagonal in the orbit basis:
\begin{equation}
 (\Phi_\sigma^U\otimes\id_R)(\ketbra{\Psi}{\Psi})
 =\sum_g\sigma(g)\ketbra{\Psi_g}{\Psi_g}.
\end{equation}
Its trace norm is $\sum_g|\sigma(g)|$.  The reverse inequality follows from the triangle inequality and $\|\Ad_{U_g}\|_\diamond=1$.

For (iv)$\Rightarrow$(iii), use the contrast $\tau$ defined in \cref{eq:tau} below.  The isometry gives $\|\Phi_\tau^U\|_\diamond=\|\tau\|_1=2$.  Norm attainment and the support-orthogonality argument in the final paragraph of this proof then produce a vector with an orthonormal $G$-orbit.

We next prove (i)$\Rightarrow$(v).  The upper bound is \cref{eq:classical-deficiency}.  Fix an arbitrary channel $\Lambda$ with Kraus operators $V_a$, and use the invariant state and purification above.  Put $P_h=\ketbra{\Psi_h}{\Psi_h}$ and $P_\perp=I-\sum_hP_h$.  Conditional on source branch $g$, the probability of orbit outcome $h$ is
\begin{align}
 K(h|g)
 &=\sum_a|\langle\Psi_h|(V_a\otimes I)|\Psi_g\rangle|^2\\
 &=\sum_a|\Tr(\rho U_h^\dagger V_aU_g)|^2\\
 &=\sum_a|\Tr(\rho U_{gh^{-1}}V_a)|^2.
 \label{eq:relative-kernel}
\end{align}
Consequently $K(h|g)=r_0(hg^{-1})$ for a nonnegative subprobability $r_0$.  Its total mass $m=\sum_xr_0(x)$ is independent of $g$; $1-m$ is leakage into $P_\perp$.

Feed the orbit seed into the two channels and measure the orbit POVM.  The target distribution is $q$ with no leakage.  The processed source distribution is $r_0*p$ with leakage $1-m$.  Measurement contractivity gives
\begin{equation}
 \frac12\|\Phi_q^U-\Lambda\Phi_p^U\|_\diamond
 \ge\frac12\bigl(\|q-r_0*p\|_1+1-m\bigr).
 \label{eq:measured-lower}
\end{equation}
Complete $r_0$ to a probability law $r=r_0+(1-m)t$ using any $t\in\Prob(G)$.  Since $t*p$ is normalized,
\begin{equation}
 \|q-r*p\|_1\le\|q-r_0*p\|_1+1-m.
\end{equation}
The right-hand side of \cref{eq:measured-lower} is therefore at least $\delta_{\rm cl}(q|p)$.  Minimizing over $\Lambda$ proves (v).

It remains to show (v)$\Rightarrow$(iii).  Let $u(g)=1/N$ and define the zero-sum contrast
\begin{equation}
 \tau(e)=1,
 \qquad
 \tau(g)=-\frac1{N-1}\quad(g\ne e).
 \label{eq:tau}
\end{equation}
For $0<t\le(N-1)/N$, the law $q_t=u+t\tau$ is nonnegative.  Since $r*u=u$ for every probability law $r$,
\begin{equation}
 \delta_{\rm cl}(q_t|u)=\TV(q_t,u)=t.
 \label{eq:classical-t}
\end{equation}
Universal equality and the admissible choice $\Lambda=\id$ imply
\begin{equation}
 t=\delta_U(q_t|u)
 \le\frac t2\|\Phi_\tau^U\|_\diamond\le t,
\end{equation}
so $\|\Phi_\tau^U\|_\diamond=2$.  In finite dimension the norm is attained by a state $\omega$ on system and ancilla.  Writing
\begin{equation}
 \overline\omega=\frac1{N-1}\sum_{g\ne e}
 (U_g\otimes I)\omega(U_g^\dagger\otimes I),
\end{equation}
we have $\|\omega-\overline\omega\|_1=2$.  Hence their supports are orthogonal.  Positivity of every summand implies
\begin{equation}
 \supp\omega\perp(U_g\otimes I)\supp\omega
 \qquad(g\ne e).
\end{equation}
Choose a unit vector $|\Psi\rangle\in\supp\omega$.  Then $\langle\Psi|(U_g\otimes I)|\Psi\rangle=0$ for every $g\ne e$, and therefore the full orbit \cref{eq:orbit} is orthonormal.  This proves (iii) and completes the equivalence.
\end{proof}

\begin{corollary}[Exact post-processing order]
\label{cor:order}
For a representation-complete carrier, where $\succeq$ denotes simulation by a single output CPTP post-processing,
\begin{equation}
 \Phi_p^U\succeq\Phi_q^U
 \quad\Longleftrightarrow\quad
 q=r*p\ \text{for some }r\in\Prob(G).
\end{equation}
No coherent or symmetry-breaking post-processing enlarges the classical convolution order.
\end{corollary}

\begin{proof}
Set the directional deficiency to zero in \Cref{thm:main-equivalence}.  The theorem identifies zero quantum deficiency exactly with membership of $q$ in the classical convolution image of $p$, which proves both implications.
\end{proof}

\begin{corollary}[Minimum carrier dimension]
\label{cor:min-dim}
The smallest carrier dimension admitting universal exact classicalization is
\begin{equation}
 d_{\rm comp}(G)=\sum_{\lambda\in\Irr(G)}d_\lambda.
 \label{eq:min-dim}
\end{equation}
It can be strictly smaller than the regular dimension $N=\sum_\lambda d_\lambda^2$.
\end{corollary}

\begin{proof}
By \Cref{thm:main-equivalence}, universal exact classicalization is equivalent to the presence of every irreducible representation type.  In the isotypic decomposition, one copy of type $\lambda$ costs $d_\lambda$ carrier dimensions, so the minimum is $\sum_\lambda d_\lambda$.
\end{proof}

\begin{theorem}[Exact ancilla costs]
\label{thm:ancilla-costs}
For the multiplicity profile \cref{eq:isotypic}, the minimum ancilla
dimension of an orthogonal orbit is
\begin{equation}
 r_{\rm orb}(U)=
 \begin{cases}
 \displaystyle\max_{\lambda\in\Irr(G)}
 \left\lceil\frac{d_\lambda}{m_\lambda}\right\rceil,
 &m_\lambda>0\ \text{for every }\lambda,\\[9pt]
 \infty,&\text{otherwise}.
 \end{cases}
 \label{eq:orbit-ancilla}
\end{equation}
If the reduced calibration state is additionally required to commute
with $U$, the minimum purification-ancilla dimension is
\begin{equation}
 r_{\rm inv}(U)=\sum_{\lambda\in\Irr(G)}d_\lambda
 \label{eq:invariant-ancilla}
\end{equation}
for every representation-complete carrier.
\end{theorem}

\begin{proof}
An orthonormal orbit spans a copy of the left regular
representation.  The multiplicity of $V_\lambda$ in
$U\otimes I_r$ is $r m_\lambda$, whereas its multiplicity in the
regular representation is $d_\lambda$.  Therefore
$rm_\lambda\ge d_\lambda$ for every $\lambda$, which is exactly the
lower bound in \cref{eq:orbit-ancilla}.  Conversely, these inequalities
allow an isometric embedding of the regular representation into
$U\otimes I_r$.  The image of the delta vector at the identity has the
required orthonormal orbit.

For the invariant problem, every commuting state has the form
\begin{equation}
 \rho=\bigoplus_\lambda I_{V_\lambda}\otimes A_\lambda,
 \qquad A_\lambda\succeq0.
 \label{eq:commuting-state}
\end{equation}
The delta-character constraint and linear independence of irreducible
characters force
\begin{equation}
 \Tr A_\lambda=\frac{d_\lambda}{N}
 \qquad(\lambda\in\Irr(G)).
 \label{eq:A-traces}
\end{equation}
Hence every $A_\lambda$ is nonzero and
\begin{equation}
 \rank\rho
 =\sum_\lambda d_\lambda\rank A_\lambda
 \ge\sum_\lambda d_\lambda.
\end{equation}
Any purification needs an ancilla at least as large as $\rank\rho$.
Choosing each $A_\lambda$ rank one, as in \cref{eq:rho}, attains the
bound.
\end{proof}

\begin{table*}[t]
\centering
\caption{Three distinct finite-group reference-frame resources.  The
regular carrier is only one way to eliminate the ancilla.}
\label{tab:resource-costs}
\small\begin{tabular}{@{}p{0.22\textwidth}p{0.44\textwidth}p{0.27\textwidth}@{}}
\toprule
Resource & Exact condition & Minimum size\\
\midrule
representation-complete carrier
 & one copy of every irreducible type
 & $\sum_\lambda d_\lambda$ carrier dimensions\\
orthogonal orbit for fixed $U$
 & $rm_\lambda\ge d_\lambda$ for all $\lambda$
 & $r_{\rm orb}(U)$ ancilla dimensions\\
invariant delta-character seed
 & complete support and \cref{eq:A-traces}
 & $\sum_\lambda d_\lambda$ ancilla dimensions\\
regular carrier
 & $m_\lambda=d_\lambda$ for all $\lambda$
 & $N=\sum_\lambda d_\lambda^2$ carrier dimensions\\
\bottomrule
\end{tabular}
\end{table*}

\section{Quantitative failure for incomplete carriers}

The converse proof gives a computable gap, not merely a logical failure.  Define
\begin{equation}
 \kappa(U)=\frac12\|\Phi_\tau^U\|_\diamond,
 \label{eq:kappa}
\end{equation}
with $\tau$ from \cref{eq:tau}, and introduce the visible Plancherel
dimension
\begin{equation}
 S(U)=\sum_{\lambda:m_\lambda>0}d_\lambda^2.
 \label{eq:visible-S}
\end{equation}
It depends only on the set of visible irreducible types, not on their
multiplicities.

\begin{theorem}[Exact twirl index and universal gap]
\label{thm:visible-index}
Let
\begin{equation}
 \mathcal E_U=\Phi_u^U=\frac1N\sum_{g\in G}\Ad_{U_g}
 \label{eq:twirl-E}
\end{equation}
be the representation twirl.  Then
\begin{align}
 \frac12\|\id-\mathcal E_U\|_\diamond
 &=1-\frac1{S(U)},\label{eq:id-twirl-distance}\\
 \kappa(U)
 &=\frac{N}{N-1}\left(1-\frac1{S(U)}\right)
 =\frac{N[S(U)-1]}{S(U)(N-1)}.
 \label{eq:kappa-exact}
\end{align}
Consequently, for every admissible $t$,
\begin{align}
 \delta_U(q_t|u)
 &\le t\,\frac{N[S(U)-1]}{S(U)(N-1)},\label{eq:general-upper}\\
 \delta_{\rm cl}(q_t|u)-\delta_U(q_t|u)
 &\ge t\,\frac{N-S(U)}{S(U)(N-1)}.
 \label{eq:explicit-gap}
\end{align}
The second bound is strict exactly when $U$ is representation
incomplete.
\end{theorem}

\begin{proof}
The isotypic formula for the twirl is
\begin{equation}
 \mathcal E_U(X)
 =\bigoplus_{\lambda:m_\lambda>0}
 \frac{I_{V_\lambda}}{d_\lambda}
 \otimes\Tr_{V_\lambda}(P_\lambda XP_\lambda).
 \label{eq:twirl-block-form}
\end{equation}
The completely bounded Pimsner--Popa inequality
\begin{equation}
 (\mathcal E_U\otimes\id_R)(X)\succeq\frac1{S(U)}X
 \qquad(X\succeq0)
 \label{eq:PP}
\end{equation}
holds for every ancilla $R$ \cite{PimsnerPopa1986}; an elementary proof is given in
\cref{app:PP}.  For any state $\omega$, it writes
\begin{equation}
 (\mathcal E_U\otimes\id)(\omega)
 =\frac{\omega}{S(U)}
 +\left(1-\frac1{S(U)}\right)\zeta
\end{equation}
for another state $\zeta$.  Therefore
\begin{equation}
 \frac12\|\omega-(\mathcal E_U\otimes\id)(\omega)\|_1
 \le1-\frac1{S(U)},
\end{equation}
which gives the upper bound in \cref{eq:id-twirl-distance}.

For the reverse bound, choose one multiplicity vector for every
visible type and define
\begin{equation}
 \rho_{\rm vis}
 =\bigoplus_{\lambda:m_\lambda>0}
 \frac{d_\lambda}{S(U)}I_{V_\lambda}
 \otimes\ketbra{a_\lambda}{a_\lambda}.
 \label{eq:visible-state}
\end{equation}
For a purification $|\Psi_{\rm vis}\rangle$, its group orbit spans the
truncated regular representation
\begin{equation}
 \mathcal R_{\rm vis}
 =\bigoplus_{\lambda:m_\lambda>0}
 V_\lambda\otimes\mathbb C^{d_\lambda}.
\end{equation}
Schur orthogonality gives
\begin{equation}
 (\mathcal E_U\otimes\id)
 (\ketbra{\Psi_{\rm vis}}{\Psi_{\rm vis}})
 =\frac1{S(U)}P_{\rm vis},
 \label{eq:visible-maxmixed}
\end{equation}
where $P_{\rm vis}$ is the projector onto that $S(U)$-dimensional
subspace.  The trace half-distance from a pure state in the range of
$P_{\rm vis}$ to $P_{\rm vis}/S(U)$ is $1-1/S(U)$, proving equality.

Finally,
\begin{equation}
 \Phi_\tau^U=\frac{N}{N-1}(\id-\mathcal E_U),
 \label{eq:tau-id-E}
\end{equation}
which proves \cref{eq:kappa-exact}.  Taking the identity converter in
\cref{eq:quantum-deficiency} gives
$\delta_U(q_t|u)\le t\kappa(U)$.  Subtracting from the classical value
$t$ yields \cref{eq:explicit-gap}.  Since
$S(U)=N$ exactly for complete irreducible support, the strictness
criterion follows.
\end{proof}

In particular, $0\le\kappa(U)\le1$, and
\begin{equation}
 \kappa(U)=1\quad\Longleftrightarrow\quad U\text{ is representation complete}.
 \label{eq:kappa-complete}
\end{equation}
Thus one fixed family of source-target laws detects every failure of
complete irreducible support and quantifies it solely through the
missing Plancherel weight.  The upper bound
\cref{eq:general-upper} need not be tight for an arbitrary reducible
carrier; the next theorem proves equality for irreducible carriers.

For irreducible carriers the gap is exact.

\begin{theorem}[Closed irreducible gap]
\label{thm:irreducible-gap}
Let $U$ be an irreducible representation of dimension $d$, and let $N=|G|$.  For $q_t=u+t\tau$,
\begin{equation}
 \delta_U(q_t|u)
 =t\,\frac{N}{N-1}\left(1-\frac1{d^2}\right),
 \qquad 0\le t\le\frac{N-1}{N}.
 \label{eq:irreducible-gap}
\end{equation}
For every nontrivial group this is strictly smaller than the classical value $t$.
\end{theorem}

\begin{proof}
Schur's lemma gives $\Phi_u^U=\mathcal D_0$, the completely depolarizing channel on dimension $d$.  Since
\begin{equation}
 \frac1{N-1}\sum_{g\ne e}\Ad_{U_g}
 =\frac{N\mathcal D_0-\id}{N-1},
\end{equation}
one obtains
\begin{equation}
 \Phi_{q_t}^U=\mathcal D_\lambda,
 \qquad
 \lambda=\frac{Nt}{N-1}.
\end{equation}
Every post-processing of $\mathcal D_0$ is a replacer channel.  Averaging its output state over all unitary conjugations cannot increase the diamond error to the unitarily covariant target, so the optimal replacer is again $\mathcal D_0$.  The standard identity
\begin{equation}
 \frac12\|\mathcal D_\lambda-\mathcal D_0\|_\diamond
 =\lambda\left(1-\frac1{d^2}\right)
 \label{eq:depol-distance}
\end{equation}
proves the formula.  A nontrivial finite group has $N>d^2$ for any one irreducible type considered in isolation, including the $d=1$ case, and the coefficient is therefore below one.
\end{proof}

\section{Composition under independent frame groups}
\label{sec:product-groups}

The visible Plancherel dimension has a simple tensor law, which makes
the classification compositional rather than merely one-shot.

\begin{theorem}[Product-group law]
\label{thm:product-law}
Let $G_1,G_2$ have representations $U^{(1)},U^{(2)}$, and let
$U=U^{(1)}\otimes U^{(2)}$ represent $G_1\times G_2$.  Then
\begin{align}
 S(U)&=S(U^{(1)})S(U^{(2)}),\label{eq:S-product}\\
 d_{\rm comp}(G_1\times G_2)
 &=d_{\rm comp}(G_1)d_{\rm comp}(G_2),\label{eq:dcomp-product}\\
 \frac12\|\id-\mathcal E_U\|_\diamond
 &=1-\frac1{S(U^{(1)})S(U^{(2)})}.
 \label{eq:product-twirl-distance}
\end{align}
The product carrier is representation complete if and only if both
factors are representation complete.
\end{theorem}

\begin{proof}
Every irreducible representation of a direct product is
$V_\lambda\otimes V_\mu$, with dimension $d_\lambda d_\mu$ and
multiplicity $m_\lambda^{(1)}m_\mu^{(2)}$ in the product carrier.
Visible pairs are exactly pairs of visible types.  Summing squared
dimensions gives
\begin{equation}
 \sum_{\lambda,\mu\ {\rm visible}}d_\lambda^2d_\mu^2
 =
 \left(\sum_{\lambda\ {\rm visible}}d_\lambda^2\right)
 \left(\sum_{\mu\ {\rm visible}}d_\mu^2\right),
\end{equation}
which proves \cref{eq:S-product} and the completeness criterion.
Summing unsquared dimensions over all types proves
\cref{eq:dcomp-product}.  The twirl is
$\mathcal E_{U^{(1)}}\otimes\mathcal E_{U^{(2)}}$, and
\cref{eq:product-twirl-distance} follows from
\cref{thm:visible-index}.
\end{proof}

\begin{corollary}[Many independent frame coordinates]
\label{cor:many-coordinates}
For the $n$-fold product group $G^n$ with carrier $U^{\otimes n}$,
\begin{align}
 S(U^{\otimes n})&=S(U)^n,\label{eq:S-power}\\
 \frac12\|\id-\mathcal E_{U}^{\otimes n}\|_\diamond
 &=1-\frac1{S(U)^n},\label{eq:twirl-power}\\
 \kappa_n(U)
 &=\frac{N^n}{N^n-1}\left(1-\frac1{S(U)^n}\right).
 \label{eq:kappa-power}
\end{align}
For an incomplete carrier, $\kappa_n(U)<1$ at every finite $n$, but
the uniform-contrast guaranteed gap is
\begin{equation}
 1-\kappa_n(U)
 =\frac{N^n-S(U)^n}{S(U)^n(N^n-1)}
 \sim S(U)^{-n}.
 \label{eq:gap-power}
\end{equation}
Thus tensoring does not restore universal classicalization, although
this particular normalized contrast becomes exponentially less
sensitive to the missing types.
\end{corollary}

\begin{proof}
Iterate \Cref{thm:product-law} over the $n$ tensor factors.  Both the visible index and the group order multiply, and substituting these products into the one-copy twirl-distance and contrast formulas gives \eqref{eq:S-power}--\eqref{eq:kappa-power}.
\end{proof}

\begin{example}[Two inequivalent $S_3$ carriers]
The irreducible dimensions of $S_3$ are $1,1,2$.  The four-dimensional carrier
\begin{equation}
 U_{\rm comp}=\mathbf1\oplus\mathrm{sgn}\oplus\mathrm{std}
\end{equation}
is representation complete and classicalizes all six-branch noise laws, despite being smaller than the six-dimensional regular representation.  The weighted state is
\begin{equation}
 \rho=\frac16\oplus\frac16\oplus\frac{2}{6}I_2.
\end{equation}
Its invariant purification has rank and ancilla dimension four.
However, \cref{eq:orbit-ancilla} shows that an arbitrary orthogonal
orbit needs only a qubit ancilla.  In a system basis adapted to
$\mathbf1\oplus\mathrm{sgn}\oplus\mathrm{std}$, one explicit seed is
\begin{equation}
 |\Psi_2\rangle
 =\frac{|{\bf1}\rangle|0\rangle
       +|\mathrm{sgn}\rangle|0\rangle}{\sqrt6}
 +\frac{|e_1\rangle|0\rangle+|e_2\rangle|1\rangle}{\sqrt3}.
 \label{eq:S3-minimal-seed}
\end{equation}
The character identity gives
$\langle\Psi_2|(U_g\otimes I_2)|\Psi_2\rangle=\delta_{g,e}$.
By contrast, the faithful two-dimensional standard representation is incomplete.  With $N=6$ and $d=2$, \cref{thm:irreducible-gap} gives
\begin{equation}
 \delta_{\rm std}(q_t|u)=\frac9{10}t,
 \qquad
 \delta_{\rm cl}(q_t|u)=t.
\end{equation}
The ten-percent gap is caused by the quantum channel encoding of the
group law and persists although the representation itself is faithful;
for this source--target pair the identity converter already attains the
optimum.
\end{example}

\begin{example}[Visible-support hierarchy for $S_3$]
For $S_3$, the visible Plancherel dimension and the universal contrast
slope \cref{eq:kappa-exact} are
\begin{center}
\footnotesize
\begin{tabular}{@{}lccc@{}}
\toprule
visible types & $S(U)$ & $\kappa(U)$ & $1-\kappa(U)$\\
\midrule
trivial only & $1$ & $0$ & $1$\\
trivial and sign & $2$ & $3/5$ & $2/5$\\
standard only & $4$ & $9/10$ & $1/10$\\
trivial and standard & $5$ & $24/25$ & $1/25$\\
all three types & $6$ & $1$ & $0$\\
\bottomrule
\end{tabular}
\end{center}
Multiplicities do not change this table.  They reduce the orbit ancilla
through \cref{eq:orbit-ancilla}, but do not restore a missing Fourier
type.
\end{example}

\begin{figure}[t]
\centering
\includegraphics[width=\columnwidth]{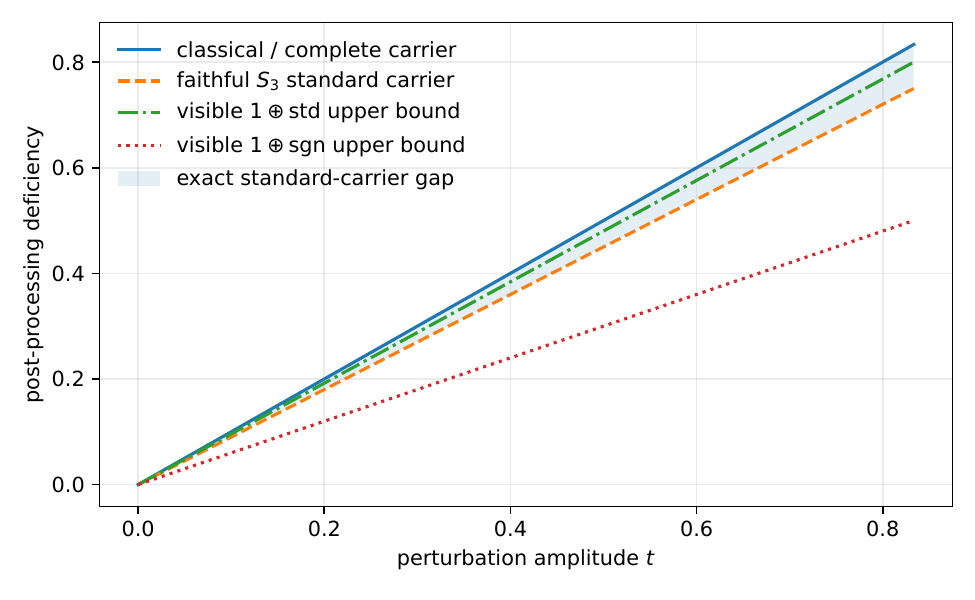}
\caption{Uniform-perturbation deficiencies on $S_3$.  The
representation-complete carrier attains the classical value; the
faithful standard carrier has exact slope $9/10$.  For the two
reducible incomplete supports, \cref{eq:general-upper} gives the shown
explicit upper bounds; equality is not asserted for those two curves.}
\label{fig:s3-gap}
\end{figure}

\section{Fourier kernels as an extreme failure mechanism}

A stronger collapse occurs when different laws induce the same channel.  Define
\begin{equation}
 \mathcal K_U^0=
 \left\{\sigma\in\mathbb R^G:
 \sum_g\sigma(g)=0,
 \ \sum_g\sigma(g)\Ad_{U_g}=0\right\}.
\end{equation}
If $0\ne\sigma\in\mathcal K_U^0$, then for sufficiently small $t>0$,
\begin{equation}
 p=u,
 \qquad q=u+t\sigma
\end{equation}
are probability laws satisfying
\begin{equation}
 \delta_U(q|u)=0,
 \qquad
 \delta_{\rm cl}(q|u)=\frac t2\|\sigma\|_1.
 \label{eq:kernel-collapse}
\end{equation}
The proof is immediate: $\Phi_q^U=\Phi_u^U$, whereas convolution leaves $u$ fixed.  This is an extreme special case of the general incomplete-carrier converse, not the only obstruction.

\begin{theorem}[Exact law-identifiability criterion]
\label{thm:law-identifiability}
Let
\begin{equation}
 V_g=U_g\otimes\overline{U_g}
 \label{eq:conjugation-representation}
\end{equation}
be the conjugation representation on vectorized operators, and let
$\mathcal T(U)$ be the set of irreducible types occurring in $V$.
Then
\begin{equation}
 p\longmapsto\Phi_p^U
\quad\text{is injective on complex coefficient functions}
\end{equation}
if and only if $V$ is representation complete.  More generally,
\begin{equation}
 \dim_{\mathbb C}\ker\!\left(
 \sigma\longmapsto\Phi_\sigma^U\right)
 =
 N-\sum_{\lambda\in\mathcal T(U)}d_\lambda^2.
 \label{eq:kernel-dimension}
\end{equation}
This criterion is strictly weaker than universal classicalization:
$U\otimes\overline U$ may be representation complete while $U$ is
not.
\end{theorem}

\begin{proof}
Vectorization identifies $\Ad_{U_g}$ with
$U_g\otimes\overline{U_g}$ up to the harmless choice of tensor order.
Thus $\Phi_\sigma^U$ is the image of the group-algebra element
\begin{equation}
 x_\sigma=\sum_{g\in G}\sigma(g)g\in\mathbb C[G]
\end{equation}
under the representation $V$.  The Wedderburn decomposition is standard \cite{Diaconis1988,Terras1999}:
\begin{equation}
 \mathbb C[G]\simeq
 \bigoplus_{\lambda\in\Irr(G)}
 \operatorname{Mat}_{d_\lambda}(\mathbb C),
 \qquad
 N=\sum_\lambda d_\lambda^2.
\end{equation}
The representation $V$ annihilates exactly the matrix-algebra blocks
whose irreducible types are absent from $V$; multiplicity of a visible
type does not change the kernel.  Therefore the image dimension is
$\sum_{\lambda\in\mathcal T(U)}d_\lambda^2$, proving
\cref{eq:kernel-dimension} and the injectivity criterion.
\end{proof}

\begin{example}[Injective laws without classicalization]
\label{ex:S3-injective}
For the standard representation of $S_3$,
\begin{equation}
 \mathrm{std}\otimes\overline{\mathrm{std}}
 \simeq\mathbf1\oplus\mathrm{sgn}\oplus\mathrm{std}.
\end{equation}
Hence all three irreducible types occur in the conjugation
representation and $p\mapsto\Phi_p^{\rm std}$ is injective.  Yet
$\mathrm{std}$ itself omits the two one-dimensional types, so
\cref{thm:main-equivalence} forbids universal classicalization and
\cref{thm:irreducible-gap} gives the exact $9/10$ slope.  This example
separates three notions that are often conflated:
\begin{equation}
 \begin{gathered}
 \text{faithful group action}
 \;\not\Rightarrow\;
 \text{universal classicalization},\\
 \text{law identifiability}
 \;\not\Rightarrow\;
 \text{universal classicalization}.
 \end{gathered}
\end{equation}
\end{example}

For the two-level phase token
\begin{equation}
 U_g=\operatorname{diag}(1,e^{2\pi ig/5}),
 \qquad g\in\mathbb Z_5,
\end{equation}
the second harmonic $\sigma(g)=\cos(4\pi g/5)$ lies in $\mathcal K_U^0$.  Channel tomography is exactly blind to that component of the phase law.

\section{Exact finite-group visibility through the conjugation Gram matrix}
\label{sec:conjugation-gram}

For a finite frame group the information retained by the reduced channel can
be read off from one explicit positive semidefinite matrix.

\begin{theorem}[Conjugation Gram identity]
\label{thm:conjugation-gram}
Let $G$ be finite, let $U:G\to\mathsf U(d)$ be a unitary representation, and
write
\[
 \Phi_p(X)=\sum_{g\in G}p_gU_gXU_g^*
\]
for a real signed coefficient vector $p$.  Define
\[
 K_{g,h}=\left|\operatorname{Tr}(U_g^*U_h)\right|^2.
\]
With the Hilbert--Schmidt norm on superoperators induced by vectorization,
\[
 \|\Phi_p-\Phi_q\|_{\mathrm{HS}}^2=(p-q)^{\mathsf T}K(p-q).
\]
Consequently the invisible signed distributions are exactly $\ker K$.  If
$K$ is positive definite, then
\begin{align}
 \|p-q\|_2
 &\leq\frac{\|\Phi_p-\Phi_q\|_{\mathrm{HS}}}
 {\sqrt{\lambda_{\min}(K)}},\\
 \frac12\|p-q\|_1
 &\leq
 \frac{\sqrt{|G|}\,\|\Phi_p-\Phi_q\|_{\mathrm{HS}}}
 {2\sqrt{\lambda_{\min}(K)}}.
\end{align}
The Euclidean constant is sharp and is attained in a smallest-eigenvalue
direction of the conjugation Gram matrix; the total-variation estimate is its
explicit dimension-dependent consequence.
\end{theorem}

\begin{proof}
Vectorization identifies conjugation by $U_g$ with
$U_g\otimes\overline U_g$.  Therefore
\[
 \Phi_p-\Phi_q=\sum_g(p_g-q_g)(U_g\otimes\overline U_g)
\]
as a matrix on the vectorized operator space.  Taking its squared
Hilbert--Schmidt norm gives the Gram form because
\[
 \operatorname{Tr}\!\bigl[(U_g\otimes\overline U_g)^*
 (U_h\otimes\overline U_h)\bigr]
 =|\operatorname{Tr}(U_g^*U_h)|^2.
\]
The kernel statement follows immediately.  If $K\succ0$, the Rayleigh bound
controls the Euclidean norm of $p-q$, and
$\|v\|_1\leq\sqrt{|G|}\|v\|_2$ completes the estimate.
\end{proof}

\begin{remark}
This identity is also a diagnostic for nonidentifiability: no amount of
statistical precision can recover a coefficient direction in $\ker K$.
Adding representation sectors changes the Gram form and can increase its rank, thereby removing blind directions.
\end{remark}

\section{Linear programs, witnesses, and finite data}

For a representation-complete carrier, \cref{eq:universal-classicalization} reduces the quantum optimization to a finite linear program.  Let $C_p$ denote convolution by $p$.  Introducing $z_h\ge0$ gives
\begin{align}
 \text{minimize}\quad &\frac12\sum_hz_h,\nonumber\\
 \text{subject to}\quad
 &-z\le q-C_pr\le z,\nonumber\\
 &r\ge0,\qquad \mathbf1^Tr=1.
 \label{eq:primal-lp}
\end{align}
A convenient dual is
\begin{align}
 \text{maximize}\quad &\frac12(y^Tq-\alpha),\nonumber\\
 \text{subject to}\quad &\|y\|_\infty\le1,\nonumber\\
 &(C_p^Ty)_x\le\alpha\quad(x\in G).
 \label{eq:dual-lp}
\end{align}
The dual vector $y$ is simultaneously a classical separating payoff and, through the orbit projectors, an ancilla-assisted quantum decision witness.

The equality also gives a sharp stability estimate.  For reconstructed laws $\widehat p,\widehat q$,
\begin{equation}
 |\delta_U(q|p)-\delta_U(\widehat q|\widehat p)|
 \le\TV(q,\widehat q)+\TV(p,\widehat p).
 \label{eq:stability}
\end{equation}
This follows from contraction of total variation under convolution.  Multinomial confidence regions can therefore be propagated without channel tomography once the branch law is directly calibrated.

\begin{corollary}[Finite-sample calibration]
\label{cor:finite-sample}
Draw $L$ independent samples from each of $p$ and $q$ and let
$\widehat p,\widehat q$ be the empirical laws.  With probability at
least $1-\gamma$,
\begin{equation}
 |\delta_U(q|p)-\delta_U(\widehat q|\widehat p)|
 \le
 N\sqrt{\frac1{2L}\log\frac{4N}{\gamma}}
 \label{eq:multinomial-bound}
\end{equation}
for every representation-complete carrier.
\end{corollary}

\begin{proof}
Hoeffding's inequality and a union bound over the $2N$ empirical
frequencies give
\begin{equation}
 \max_g\bigl\{
 |\widehat p(g)-p(g)|,\,
 |\widehat q(g)-q(g)|
 \bigr\}
 \le\sqrt{\frac1{2L}\log\frac{4N}{\gamma}}
\end{equation}
with the stated confidence.  Each total-variation error is at most
$N/2$ times this coordinate bound.  Apply \cref{eq:stability}.
\end{proof}

The bound is deliberately distribution-free and conservative.
Likelihood-ratio or exact multinomial regions can be inserted into the
same stability inequality without changing the comparison theorem.

\section{Finite-dimensional obstruction for infinite compact groups}
\label{sec:compact-obstruction}

The finite-group result does not extend without qualification.  Without an energy or bandwidth constraint, no finite-dimensional carrier can encode every probability law on an infinite compact group isometrically in total variation.

\begin{theorem}[Noninjectivity for all finite-dimensional compact-group carriers]
\label{thm:compact-noninjectivity}
Let $K$ be an infinite compact group and $U:K\to\mathcal U(\mathcal H)$ a
finite-dimensional unitary representation, $\dim\mathcal H=d$.  Define
\[
 \Phi_\mu(X)=\int_K U_gXU_g^\dagger\,\mathrm d\mu(g)
\]
for Borel probability measures $\mu$.  Then there exist distinct, mutually
singular probability measures $p,q$ of finite support such that
\begin{equation}
 \Phi_p=\Phi_q,
 \qquad d_{\rm TV}(p,q)=1.
 \label{eq:compact-kernel}
\end{equation}
Consequently $\mu\mapsto\Phi_\mu$ is never injective on all Borel laws, and
no positive constant $c$ can satisfy
$\frac12\|\Phi_p-\Phi_q\|_\diamond\geq c\,d_{\rm TV}(p,q)$ universally.
In particular, the exact finite-group classicalization isometry cannot extend
verbatim to an infinite compact group in finite dimension.
\end{theorem}

\begin{proof}
The real vector space of linear maps on the $d^2$-dimensional real space of
Hermitian matrices has dimension $d^4$.  Choose $d^4+1$ distinct elements
$g_j\in K$.  The $d^4+1$ real-linear maps $\operatorname{Ad}_{U_{g_j}}$ lie in a
$d^4$-dimensional real vector space, so there are real coefficients, not all
zero, with
\[
 \sum_j c_j\operatorname{Ad}_{U_{g_j}}=0.
\]
Applying this identity to the unit gives $\sum_jc_j=0$.  Hence the positive
and negative parts have the same nonzero mass $C$.  Set
\[
 p=C^{-1}\sum_jc_j^+\delta_{g_j},\qquad
 q=C^{-1}\sum_jc_j^-\delta_{g_j}.
\]
Their supports are disjoint, so $d_{\rm TV}(p,q)=1$, while the linear
dependence gives $\Phi_p=\Phi_q$.
\end{proof}

\begin{theorem}[Stable reconstruction on every visible Peter--Weyl band]
\label{thm:bandlimited-compact-stability}
Let $K$ be a compact group with normalized Haar measure, let
$U:K\to\mathcal U(\mathcal H)$ be finite dimensional, and decompose the
conjugation representation $V=U\otimes\overline U$ as
\[
 V\simeq\bigoplus_{\lambda\in\widehat K}
 U_\lambda\otimes I_{m_\lambda}.
\]
Choose a finite set $\Lambda_0\subseteq\{\lambda:m_\lambda>0\}$ and let
$\mathcal E_{\Lambda_0}$ be the real Peter--Weyl space of integrable
functions whose Fourier support is contained in $\Lambda_0$ and their
conjugates.  For
$\Phi_\sigma(X)=\int_K\sigma(g)U_gXU_g^\dagger\,\mathrm dg$, the restriction
$\sigma\mapsto\Phi_\sigma$ to $\mathcal E_{\Lambda_0}$ is injective.  The
optimal constant
\begin{equation}
 c_{U,\Lambda_0}:=
 \inf_{\substack{\sigma\in\mathcal E_{\Lambda_0}\\
                  \|\sigma\|_{L^1}=1}}
 \|\Phi_\sigma\|_\diamond
 \label{eq:bandlimited-constant}
\end{equation}
is strictly positive, and
\begin{equation}
 c_{U,\Lambda_0}\|\sigma\|_{L^1}
 \leq\|\Phi_\sigma\|_\diamond
 \leq\|\sigma\|_{L^1}.
 \label{eq:bandlimited-stability}
\end{equation}
Thus for probability densities $p,q$ with $p-q\in\mathcal E_{\Lambda_0}$,
\begin{equation}
 c_{U,\Lambda_0}d_{\rm TV}(p,q)
 \leq\frac12\|\Phi_p-\Phi_q\|_\diamond
 \leq d_{\rm TV}(p,q).
 \label{eq:bandlimited-tv-diamond}
\end{equation}
Within a prescribed finite visible band, the noninjectivity obstruction of Theorem~\ref{thm:compact-noninjectivity} disappears and inversion is quantitatively stable.  Components outside the visible band remain unconstrained.
\end{theorem}

\begin{proof}
In a basis adapted to the displayed decomposition, Peter--Weyl
orthogonality gives
\[
 \int_K\sigma(g)V_g\,\mathrm dg
 \simeq\bigoplus_{\lambda:m_\lambda>0}
 \widehat\sigma(\lambda)\otimes I_{m_\lambda},
\]
up to the harmless transpose convention in the Fourier transform.  If
$\sigma\in\mathcal E_{\Lambda_0}$ and $\Phi_\sigma=0$, every visible Fourier
matrix $\widehat\sigma(\lambda)$ vanishes; hence $\sigma=0$.  The $L^1$ unit
sphere in the finite-dimensional space $\mathcal E_{\Lambda_0}$ is compact,
and the continuous function $\sigma\mapsto\|\Phi_\sigma\|_\diamond$ has no
zero on it.  Its minimum is therefore the positive number
\eqref{eq:bandlimited-constant}.  The upper bound follows from the triangle
inequality and $\|\operatorname{Ad}_{U_g}\|_\diamond=1$.  Applying the result
to $\sigma=p-q$, with $\|p-q\|_1=2d_{\rm TV}(p,q)$, proves
\eqref{eq:bandlimited-tv-diamond}.
\end{proof}

\section{An explicit condition number for visible compact-group bands}

The compactness proof of Theorem~\ref{thm:bandlimited-compact-stability}
shows positivity of the inverse constant but does not display it.  A direct
Peter--Weyl estimate gives a universal computable lower bound.

\begin{theorem}[Explicit bandlimited diamond stability]
\label{thm:explicit-bandlimited-constant}
Use the notation of Theorem~\ref{thm:bandlimited-compact-stability}, let
$D=\dim\mathcal H$, and put
\begin{equation}
 S_{\Lambda_0}=\sum_{\lambda\in\Lambda_0}d_\lambda^2.
 \label{eq:band-plancherel-size}
\end{equation}
Then every real $\sigma\in\mathcal E_{\Lambda_0}$ satisfies
\begin{equation}
 \boxed{\qquad
 \|\Phi_\sigma\|_\diamond
 \geq\frac{1}{\sqrt D\,S_{\Lambda_0}}
       \|\sigma\|_{L^1(K)}.
 \qquad}
 \label{eq:explicit-bandlimited-stability}
\end{equation}
In particular,
$c_{U,\Lambda_0}\geq(\sqrt D\,S_{\Lambda_0})^{-1}$ and, for probability
densities with bandlimited difference,
\begin{equation}
 \frac{1}{\sqrt D\,S_{\Lambda_0}}d_{\rm TV}(p,q)
 \leq\frac12\|\Phi_p-\Phi_q\|_\diamond.
 \label{eq:explicit-bandlimited-tv}
\end{equation}
The estimate is independent of invisible multiplicities and depends only on
the carrier dimension and the total visible Plancherel weight of the chosen
band.
\end{theorem}

\begin{proof}
Peter--Weyl inversion and $\|U_\lambda(g)\|_{\rm op}=1$ give
\[
 \|\sigma\|_{L^1}
 \leq\sum_{\lambda\in\Lambda_0}
 d_\lambda\|\widehat\sigma(\lambda)\|_1
 \leq S_{\Lambda_0}
 \max_{\lambda\in\Lambda_0}
 \|\widehat\sigma(\lambda)\|_{\rm op}.
\]
On Hilbert--Schmidt operator space, the Liouville matrix of $\Phi_\sigma$ is,
up to transpose conventions,
$\bigoplus_\lambda\widehat\sigma(\lambda)\otimes I_{m_\lambda}$.
Hence its $2\to2$ norm equals the maximum on the right.  Choose an operator
$X$ of Hilbert--Schmidt norm one attaining that norm.  Since
$\|X\|_1\leq\sqrt D\|X\|_2$ and
$\|\Phi_\sigma(X)\|_1\geq\|\Phi_\sigma(X)\|_2$,
\[
 \|\Phi_\sigma\|_\diamond
 \geq\|\Phi_\sigma\|_{1\to1}
 \geq\frac1{\sqrt D}
 \max_\lambda\|\widehat\sigma(\lambda)\|_{\rm op}.
\]
Combining the two inequalities proves
\eqref{eq:explicit-bandlimited-stability}; applying it to $p-q$ gives
\eqref{eq:explicit-bandlimited-tv}.
\end{proof}

\section{Physical interpretation and scope}

The theorem distinguishes three carrier properties.

First, \emph{regularity} supplies each irreducible with multiplicity $d_\lambda$ and stores a literal group register.  It is sufficient but generally wasteful.

Second, \emph{representation completeness} stores one copy of each symmetry type.  This is exactly the threshold for universal classicalization.  A single weighted calibration state turns the reduced channel family into a faithful classical statistical experiment under arbitrary post-processing.

Third, ordinary \emph{faithfulness} of $g\mapsto U_g$ is weaker.  The standard representation of $S_3$ is faithful but does not classicalize degradation.  The missing resource is not knowledge of the abstract group element at the representation level; it is one input whose entire orbit is perfectly distinguishable.

Multiplicity has two sharply different roles.  It cannot compensate for
a missing irreducible type and therefore does not change $S(U)$ or the
universal contrast \cref{eq:kappa-exact}.  It can, however, reduce the
ancilla needed to embed a regular orbit, exactly as
\cref{eq:orbit-ancilla} quantifies.  The invariant calibration state
forgets this advantage because its forced block ranks give the larger
cost \cref{eq:invariant-ancilla}.

The result is finite-group and finite-dimensional.  For compact continuous groups an exactly orthogonal orbit indexed by every group element is nonnormalizable in finite dimension.  Energy constraints, approximate designs, and weak topologies are therefore essential, and the finite theorem should not be extrapolated by replacing sums with integrals.  Projective representations require passage to the effective conjugation group because central phases label the same channel.

\input{additional_results}
\input{frontier_results}
\input{final_results}

\section{Conclusion}

Universal classicalization of finite group-valued reference noise has an exact representation-theoretic boundary.  It holds precisely when the carrier contains every irreducible type, equivalently when one ancilla-assisted input resolves the full group orbit.  Under this condition arbitrary coherent post-processing gives no advantage beyond classical convolution.  If an irreducible type is missing, a fixed perturbation of the uniform law yields a strict quantum--classical gap controlled by the missing Plancherel weight; for irreducible carriers the corresponding deficiency is analytic.  The exact orbit and invariant-seed ancilla costs complete the finite-group resource accounting.

The conjugation representation separately determines whether the underlying law is identifiable from the induced channel, and the conjugation Gram matrix gives a direct conditioning certificate.  These criteria are weaker than universal classicalization, as the standard representation of $S_3$ demonstrates.  For infinite compact groups, no finite-dimensional carrier can identify all probability laws, but every prescribed visible Peter--Weyl band is stably reconstructible with an explicit condition number.  The resulting hierarchy separates group-action faithfulness, law identifiability, representation completeness, and the ancillary resources needed for operational branch resolution.

\appendix

\section{Completely bounded Pimsner--Popa inequality}
\label{app:PP}

We give a finite-dimensional proof of \cref{eq:PP} that also explains
the constant $S(U)$.  First, for every positive operator
$Y\succeq0$ on $\mathbb C^d\otimes K$,
\begin{equation}
 \frac{I_d}{d}\otimes\Tr_{\mathbb C^d}Y
 \succeq\frac1{d^2}Y.
 \label{eq:local-depol-PP}
\end{equation}
It is enough by spectral decomposition to take
$Y=|\psi\rangle\langle\psi|$.  In a Schmidt basis,
\begin{equation}
 |\psi\rangle=\sum_{i=1}^r\sqrt{\lambda_i}|i\rangle|i'\rangle,
 \qquad r\le d.
\end{equation}
For an arbitrary vector $|\phi\rangle$, Cauchy--Schwarz gives
\begin{align}
 |\langle\psi|\phi\rangle|^2
 &\le r\sum_{i=1}^r\lambda_i
 |\langle i,i'|\phi\rangle|^2\\
 &\le d^2\left\langle\phi\left|
 \frac{I_d}{d}\otimes
 \sum_i\lambda_i|i'\rangle\langle i'|
 \right|\phi\right\rangle,
\end{align}
which is \cref{eq:local-depol-PP}.

We also need a weighted pinching lemma.  Let
$\mathcal K=\bigoplus_\lambda\mathcal K_\lambda$,
$c_\lambda>0$, and suppose
\begin{equation}
 A_\lambda\succeq
 \frac1{c_\lambda}|\psi_\lambda\rangle\langle\psi_\lambda|.
\end{equation}
Then, for $|\psi\rangle=\bigoplus_\lambda|\psi_\lambda\rangle$,
\begin{equation}
 \bigoplus_\lambda A_\lambda
 \succeq
 \frac1{\sum_\lambda c_\lambda}
 |\psi\rangle\langle\psi|.
 \label{eq:weighted-pinching}
\end{equation}
Indeed, with
$a_\lambda=\langle\phi_\lambda|A_\lambda|\phi_\lambda\rangle$,
\begin{equation}
 |\langle\psi|\phi\rangle|
 \le\sum_\lambda\sqrt{c_\lambda a_\lambda}
 \le\sqrt{\sum_\lambda c_\lambda}\,
    \sqrt{\sum_\lambda a_\lambda}.
\end{equation}

Now apply \cref{eq:local-depol-PP} to every isotypic diagonal block of
a rank-one positive operator $X=|\psi\rangle\langle\psi|$, with
$d=d_\lambda$ and with the multiplicity space and arbitrary external
ancilla included in $K$.  Equation \cref{eq:weighted-pinching} with
$c_\lambda=d_\lambda^2$ yields
\begin{equation}
 (\mathcal E_U\otimes\id)(X)
 \succeq
 \frac{X}{\sum_{\lambda:m_\lambda>0}d_\lambda^2}.
\end{equation}
Spectral decomposition extends this to all $X\succeq0$, proving
\cref{eq:PP}.  The state \cref{eq:visible-state} saturates the induced
trace-distance bound, so the constant cannot be improved in the
completely bounded order.

\section{Depolarizing-channel diamond distance}

For the proof of \cref{thm:irreducible-gap}, let
$\mathcal D_\lambda=\lambda\id+(1-\lambda)\mathcal D_0$.  Then
\begin{equation}
 \|\mathcal D_\lambda-\mathcal D_0\|_\diamond
 =|\lambda|\|\id-\mathcal D_0\|_\diamond
 =2|\lambda|\left(1-\frac1{d^2}\right).
\end{equation}
The lower bound is attained on a maximally entangled state.  The matching upper bound follows from the standard semidefinite program for the diamond norm and unitary covariance; see \cite{Watrous2018}.  We record the identity to make the normalization in \cref{eq:irreducible-gap} explicit.

\section{Reproducibility}

The ancillary script deterministically regenerates \cref{fig:s3-gap} and checks the finite-group examples, the CPTP-induced convolution kernel, the visible-index witness, and \cref{eq:PP}.  It writes no report unless \texttt{--write-report} is supplied.  These calculations are regression tests; the stated constants and equivalences are proved analytically.

\section*{Scope, limitations, and open problems}

The manuscript proves an if-and-only-if finite-group boundary: universal
classicalization holds exactly when the carrier contains every irreducible
type, with explicit witnesses when a type is missing.  The ancillary
program checks character decompositions, orbit resolution, convolution
channels, and the counterexamples over representative finite groups.  The
result is exact for the declared finite-dimensional finite-group setting.

The subsequent results provide two approximate layers: a finite-cutoff
compact-group reconstruction bound, and a finite-group orbit-frame
certificate in which the smallest Gram eigenvalue controls every signed
branch contrast simultaneously.  The remaining target
is the full resource-cost law for coherent post-processing.  It must match
the dependence on the band condition number with lower witnesses and
separate missing Plancherel weight from poor frame conditioning.

\section*{finite-cutoff reconstruction bound}

Let \(P_E\) be a finite energy cutoff and let \(T_E\) be the truncated frame
analysis map on the retained representation space.  If
\(T_E^*T_E\ge\lambda_E P_E\) with \(\lambda_E>0\), then every retained
operator is reconstructed from its frame coefficients with inverse norm at
most \(\lambda_E^{-1/2}\).  A missing retained coefficient vector of norm
\(\eta\) can therefore change the reconstructed operator by at most
\(\eta/\sqrt{\lambda_E}\).

If a state has tail weight
\(\delta_E=\operatorname{tr}[(1-P_E)\rho]\), projection and the gentle
measurement estimate add at most \(2\sqrt{\delta_E}\) in trace norm.
Thus finite-group classicalization bounds extend to a compact-group
cutoff with the explicit error budget
\[
 \frac{\eta}{\sqrt{\lambda_E}}+2\sqrt{\delta_E}.
\]
Missing-irrep witnesses give the matching obstruction whenever
\(\lambda_E=0\) on an occupied retained sector.

The remaining compact-group theorem is to control
\(\lambda_E\), the missing Plancherel weight \(\eta\), and the energy tail
\(\delta_E\) uniformly as \(E\to\infty\).

\input{extended_results}

\end{document}

%% file: physics_format.tex
\makeatletter
\AtBeginDocument{%
  \@ifpackageloaded{hyperref}{\hypersetup{hidelinks}}{}%
}
\makeatother

%% file: additional_results.tex
\section*{frame-conditioned robustness}

The exact irrep criterion admits a quantitative finite-resolution version
controlled by the frame operator.

\begin{theorem}[Stable classicalization on a resolved support]
\label{thm:third-classicalization-frame}
Let \(S\) be the frame operator of the resolved representation sector and
assume its smallest eigenvalue on that support is \(\sigma>0\).  If an
implemented frame has operator \(S+E\) with \(\|E\|<\sigma\), then
\[
 \begin{aligned}
 \|(S+E)^{-1}\|
 &\le\frac1{\sigma-\|E\|},\\
 \|(S+E)^{-1}-S^{-1}\|
 &\le
 \frac{\|E\|}{\sigma(\sigma-\|E\|)}.
 \end{aligned}
\]
Hence a reconstruction map whose remaining synthesis operator has norm
at most \(C\) changes by at most
\[
 \frac{C\|E\|}{\sigma(\sigma-\|E\|)}.
\]
Every strict classicalization equality or coherent-advantage witness with
diamond-norm margin \(m\) survives whenever the induced channel error is
smaller than \(m\).
\end{theorem}

\begin{proof}
The inverse bounds follow from the resolvent identity and the Neumann
series on the supported subspace.  The witness statement follows from
the Lipschitz continuity of every norm-one linear separation functional.
\end{proof}

The theorem distinguishes missing representation content from poor
conditioning.  A missing irrep gives an exact structural obstruction;
a present but weakly resolved irrep gives a large factor
\((\sigma-\|E\|)^{-1}\).  Compact-group cutoff limits therefore require
both vanishing omitted Plancherel weight and a controlled lower frame
bound.

%% file: frontier_results.tex
\section*{subgroup inheritance of universal classicalization}

The exact representation criterion is monotone under restriction to a
physical subgroup.

\begin{theorem}[Subgroup inheritance]
\label{thm:frontier-classicalization-subgroup}
Let \(H\le G\) be finite groups.  If a \(G\)-carrier \(U\) contains every
irreducible \(G\)-type, then its restricted carrier
\(\operatorname{Res}^G_HU\) contains every irreducible \(H\)-type.
Consequently universal classicalization for \(G\) implies universal
classicalization for every subgroup noise model, and
\[
 d_{\rm comp}(G)\ge d_{\rm comp}(H).
\]
\end{theorem}

\begin{proof}
Fix \(\tau\in\operatorname{Irr}(H)\).  The induced representation
\(\operatorname{Ind}_H^G\tau\) contains some irreducible
\(\lambda\in\operatorname{Irr}(G)\).  Frobenius reciprocity gives
\[
 \operatorname{Hom}_H
 \!\left(\tau,\operatorname{Res}^G_H\lambda\right)
 \cong
 \operatorname{Hom}_G
 \!\left(\operatorname{Ind}_H^G\tau,\lambda\right)\ne0.
\]
Since \(U\) contains \(\lambda\), its restriction contains \(\tau\).
Thus the restricted carrier is representation complete and the main
equivalence applies.  Restricting a minimum complete \(G\)-carrier gives
an \(H\)-complete carrier of the same dimension, proving the inequality.
\end{proof}

This supplies an immediate obstruction test: failure on one experimentally
accessible subgroup rules out universal classicalization for the full
symmetry, without optimizing over all \(G\)-valued noise laws.

%% file: final_results.tex
\section*{Subquotient obstruction hierarchy}

Universal classicalization descends not only to subgroups but to every
physical subquotient.

\begin{theorem}[Subquotient inheritance]
\label{thm:final-classicalization-subquotient}
Let \(H\le G\) and \(N\trianglelefteq H\).  If a \(G\)-carrier \(U\)
contains every irreducible \(G\)-type, then the \(N\)-fixed subspace of
\(\operatorname{Res}^G_HU\) contains every irreducible representation of
\(H/N\).  Consequently
\[
 d_{\rm comp}(G)\ge d_{\rm comp}(H/N)
\]
for every subquotient \(H/N\) of \(G\).
\end{theorem}

\begin{proof}
Every irreducible \(\tau\) of \(H/N\) inflates to an irreducible
\(\widetilde\tau\) of \(H\) on which \(N\) acts trivially.
Theorem~\ref{thm:frontier-classicalization-subgroup} shows that
\(\operatorname{Res}^G_HU\) contains \(\widetilde\tau\).  Its carrier lies
inside the \(N\)-fixed subspace and realizes \(\tau\).  Applying this to
a minimum complete \(G\)-carrier proves the dimension inequality.
\end{proof}

Hence every experimentally accessible reduced symmetry supplies a lower
bound on the coherent resource required by the full group.  The strongest
subquotient obstruction can be computed before optimizing over
\(G\)-valued noise.

\section*{Robust branch visibility from an imperfect orbit frame}

Exact representation completeness provides an orthonormal orbit seed.  In an
experiment the prepared seed may be imperfect, so the relevant quantity is
the lower frame bound of its orbit rather than exact orthogonality.

Let \(|\Psi_g\rangle=(U_g\otimes I_R)|\Psi\rangle\) and let
\(G_\Psi\) be the \(|G|\times|G|\) Gram matrix
\[
 (G_\Psi)_{g,h}=\langle\Psi_g|\Psi_h\rangle.
\]

\begin{theorem}[Approximate orbit-isometry certificate]
\label{thm:approximate-orbit-isometry}
Assume
\[
 G_\Psi\succeq a I_{|G|}
\]
for some \(a>0\).  Then every complex signed group law \(\sigma\) satisfies
\[
 a\|\sigma\|_1
 \le
 \|\Phi_\sigma^U\|_\diamond
 \le
 \|\sigma\|_1.
\]
In particular, for probability laws \(p,q\),
\[
 a\,\|p-q\|_{\rm TV}
 \le
 \frac12\|\Phi_p^U-\Phi_q^U\|_\diamond
 \le
 \|p-q\|_{\rm TV}.
\]
If \(\|G_\Psi-I\|_{\rm op}\le\epsilon<1\), one may take
\(a=1-\epsilon\).
\end{theorem}

\begin{proof}
Let \(V:\ell_2(G)\to\mathcal H\otimes R\) be the synthesis map
\(V|g\rangle=|\Psi_g\rangle\), so \(V^*V=G_\Psi\).  On the seed state,
\[
 (\Phi_\sigma^U\otimes\id_R)(|\Psi\rangle\langle\Psi|)
 =V D_\sigma V^*,
\]
where \(D_\sigma\) is diagonal with entries \(\sigma(g)\).  Write the polar
decomposition \(V=W G_\Psi^{1/2}\), with \(W\) an isometry.  Since
\(G_\Psi\succeq aI\),
\[
 \|D_\sigma\|_1
 \le
 \|G_\Psi^{-1/2}\|_{\rm op}^2
 \|G_\Psi^{1/2}D_\sigma G_\Psi^{1/2}\|_1
 \le
 a^{-1}\|V D_\sigma V^*\|_1.
\]
This gives the lower diamond-norm bound.  The upper bound is the triangle
inequality together with \(\|\operatorname{Ad}_{U_g}\|_\diamond=1\).
The probability-law statement follows by taking \(\sigma=p-q\), and the
last claim follows from the smallest-eigenvalue bound.
\end{proof}

Thus exact orbit orthogonality is not a brittle all-or-nothing certificate.
The smallest Gram eigenvalue gives a directly measurable multiplicative
retention factor for every branch contrast simultaneously.  A vanishing
lower frame bound recovers the missing-irrep obstruction, while a well
conditioned orbit certifies robust distinguishability before any
post-processing optimization is attempted.

%% file: extended_results.tex
\section*{Sobolev compact-group law}

The visible-band theorem extends from exactly bandlimited laws to smooth
laws with a quantitative, vanishing tail.  This supplies an approximate
compact-group identifiability theorem without assuming that sums may
simply be replaced by integrals.

\begin{theorem}[Sobolev-stable compact-group reconstruction]
\label{thm:fourth-sobolev-compact}
Let \(K\) be a compact Lie group with normalized Haar measure and positive
Laplace--Beltrami operator \(\mathcal L\).  Let \(\Pi_E\) be the spectral
projector onto eigenvalues at most \(E\), and assume that every irreducible
type occurring in \(\operatorname{ran}\Pi_E\) is visible in
\(U\otimes\overline U\).  Write \(c_E>0\) for the bandlimited stability
constant in \cref{eq:bandlimited-constant}.  If \(p,q\) are probability
densities and \(\sigma=p-q\in H^s(K)\), then, with
\[
 D_U(p,q)=\frac12\|\Phi_p-\Phi_q\|_\diamond ,
\]
one has
\begin{equation}
 \boxed{\quad
 \TV(p,q)
 \le
 \frac{D_U(p,q)}{c_E}
 \frac{1+c_E^{-1}}{2}(1+E)^{-s/2}
 \|\sigma\|_{H^s}.
 \quad}
\label{eq:fourth-sobolev-upper}
\end{equation}
Equivalently,
\begin{equation}
 D_U(p,q)
 \ge
 c_E\TV(p,q)
 -\frac{1+c_E}{2}(1+E)^{-s/2}
 \|\sigma\|_{H^s}.
\label{eq:fourth-sobolev-lower}
\end{equation}
\end{theorem}

\begin{proof}
Spectral calculus and normalized Haar measure give
\[
 \begin{aligned}
 t_E&:=\|(I-\Pi_E)\sigma\|_{L^1}\\
 &\le\|(I-\Pi_E)\sigma\|_{L^2}\\
 &\le(1+E)^{-s/2}\|\sigma\|_{H^s}.
 \end{aligned}
\]
Apply the visible-band lower bound to \(\Pi_E\sigma\).  Contractivity of
the integrated conjugation channel gives
\[
 \begin{aligned}
 \|\sigma\|_{L^1}
 &\le \|\Pi_E\sigma\|_{L^1}+t_E\\
 &\le c_E^{-1}\|\Phi_{\Pi_E\sigma}\|_\diamond+t_E\\
 &\le c_E^{-1}\|\Phi_\sigma\|_\diamond
      +(1+c_E^{-1})t_E.
 \end{aligned}
\]
Divide by two and substitute the tail estimate.  Rearrangement gives
\eqref{eq:fourth-sobolev-lower}.
\end{proof}

For a family with
\(\sup_E\|p_E-q_E\|_{H^s}\le B\), the only two obstructions are now
explicit: deterioration of \(c_E\) and the tail
\((1+E)^{-s/2}B\).  If a carrier sequence makes
\(c_E\) no smaller than an inverse polynomial while the cutoff grows,
the reconstruction error vanishes at a computable rate.  The remaining
resource-theoretic problem is sharper: determine whether the same rate
controls the advantage of arbitrary coherent post-processing, and
construct lower witnesses matching the dependence on \(c_E\).